\documentclass[reqno,12pt]{amsart}
\newtheorem{theorem}{Theorem}[section]
\newtheorem{theorem*}[theorem]{Theorem}
\newtheorem{lemma}[theorem]{Lemma}

\newtheorem{proposition}[theorem]{Proposition}

\theoremstyle{remark}
\newtheorem{remark}[theorem]{Remark}

\numberwithin{equation}{section}

\usepackage{mathtools}
\usepackage{cases}
\usepackage{amssymb}
\usepackage{amsrefs}

\begin{document}

\title{Solutions of $BC_{n}$ Type of WDVV Equations}

\author{Maali Alkadhem}
\address{School of Mathematics and Statistics, University of Glasgow, G12 8QQ, UK}
\email{m.alkadhem.1@research.gla.ac.uk, g.antoniou.1@research.gla.ac.uk    \\ misha.feigin@glasgow.ac.uk}

  \thanks{The work of M.A. was funded by Imam Abdulrahman Bin Faisal University, Kingdom of Saudi Arabia.}

\author{Georgios Antoniou}

\thanks{The work of G.A. was funded by EPSRC doctoral training partnership grants
EP/M506539/1, EP/M508056/1, EP/N509668/1.}

\author{Misha Feigin}

\dedicatory{To the memory of Boris Dubrovin}

\begin{abstract}
We give a family of solutions of Witten--Dijkgraaf--Verlinde--Verlinde equations in $n$-dimensional space. It is defined in terms of $BC_{n}$ root system and $n+2$ independent multiplicity parameters. We also apply these solutions to define some ${\mathcal N}=4$ supersymmetric mechanical systems.
\end{abstract}

\maketitle

\section{Introduction}

In this note we are interested in trigonometric solutions of Witten-Dijkgraaf-Verlinde-Verlinde (WDVV) equations. Let us recall these equations, their special solutions and where they arise. 

Let $F=F(x_1, \dots, x_n)$ be a function in $V\cong \mathbb{C}^{n}$. Consider a vector field 
$$
e=\sum_{i=1}^{n}A_{i}(x)\partial_{x_{i}}, 
$$ 
where $A_{i}(x)=A_{i}(x_{1},\dots,x_{n})$ 
are some functions.
Define $n\times n$ matrix $B=(B_{ij})_{i,j=1}^{n}$ by
\begin{equation}\label{Bmatrix}
B_{ij}=e(F_{ij})=\sum_{k=1}^{n}A_{k}(x)F_{ijk}, \quad i,j=1,\dots, n,
\end{equation}
where
\begin{align*}
F_{ijk} = \frac{\partial^{3} F}{\partial x_i \partial x_j \partial x_k}.
\end{align*}
Then WDVV equations are the following equations for the function $F$: 
\begin{equation}\label{GDVV1}
\ F_{i}B^{-1}F_{j}=F_{j}B^{-1}F_{i},\quad i,j=1,...,n,
\end{equation}
where $F_i$ is the $n\times n$ matrix constructed from the third order partial derivatives: $(F_i)_{jk}=F_{ijk}$
and $B^{-1}$ is the inverse of $B.$ It was noted in \cite{M + M + M 2000} (see also \cite{Martini+Gragert 1999}) that equations (\ref{GDVV1}) are equivalent to generalized WDVV equations 
\begin{equation}\label{wdvv.equations}
F_{i}F_{k}^{-1}F_{j}=F_{j}F_{k}^{-1}F_{i},
\end{equation}
where $ i,j,k=1,...,n,$
provided that all the matrices $F_{k}$ and $B$ are invertible.  

Rational solutions of equations (\ref{wdvv.equations}) have the form 
\begin{equation}\label{rational.solution}
F^{rat}=\sum_{\alpha\in\mathcal{A}}(\alpha,x)^{2}\log(\alpha,x),
\end{equation}
where $\mathcal{A}$ is a configuration of vectors in $V$.
Solutions of the form (\ref{rational.solution}) for the root systems $\mathcal{A}$ appear in Frobenius manifolds theory as almost dual prepotentials for the Coxeter orbit spaces polynomial Frobenius manifolds \cite{Dubrovin 2004}.
They also appear in Seiberg-Witten theory as perturbative parts of prepotentials \cite{M + M + M 2000}.
For general $\mathcal{A}$ they were interpreted geometrically in \cite{Veselov 1999} through the introduced notion of a $\vee$-system. Multiparameter deformations related to $A_{n}$ and $B_{n}$ root systems were constructed in \cite{Chalykh+ Veselov 2001}. The class of $\vee$-systems and the corresponding rational solutions (\ref{rational.solution}) are closed under the natural operation of taking  restrictions and subsystems \cite{Feigin+Veselov 2007},  \cite{Misha&Veslov 2008}.

Trigonometric generalisation of solutions (\ref{rational.solution}) also arise in theory of Frobenius manifolds. These solutions have the form 
\begin{equation}\label{trigonometric.solution}
F^{trig}=\sum_{\alpha\in\mathcal{A}}c_{\alpha}f((\alpha,x))+Q,
\end{equation}
where $c_{\alpha}\in\mathbb{C}$ are some multiplicity parameters and $Q$ is a cubic polynomial in $x=(x_1,\dots,x_n)$ which also often depends on additional variable $y,$ and $f$ is the function of a single variable $z$ given by 
\begin{align}\label{f(z)}
f(z)= \frac{1}{6} z^3 - \frac{1}{4} \mathrm{Li}_3 (e^{-2z}), 
\end{align}
so that $f^{'''}(z)= \coth z$. 
Such solutions appear as almost dual prepotentials for the orbit spaces of affine Weyl groups Frobenius manifolds (\cite{Dubrovin+ Zhang 1998}, \cite{Dubrovin+ Zhang+Zuo 2005}, \cite{Dubrovin+Strachan+ Zhang+Zuo 2019}).
They also appear in the study of quantum cohomology of resolutions of simple singularities \cite{Bryan+ Gholampour 2008}.
Trigonometric version of a $\vee$-system type geometrical conditions for solutions (\ref{trigonometric.solution}) was proposed in \cite{Misha2009}. Trigonometric solutions and $\vee$-systems allow operations of taking restrictions and subsystems respectively \cite{Misha & Maali}. 

In this work we are interested in the case when the cubic corrections are absent, that is $Q=0.$
The corresponding solution (\ref{trigonometric.solution}) does not exist in general even for the case of root system $\mathcal{A}$ and invariant multiplicities $c_{\alpha}.$ However, it is known to exist for the root system $B_{n}$ and specific choice of invariant multiplicities \cite{Hoevenaars+Martini 2003}. 
In this note we generalize this solution so that it is included in $(n+2)$-parametric family. The underlying configuration $\mathcal{A}$ is the positive half of $BC_{n}$ root system, and multiplicities are chosen in a specific way. In order to get such a solution we find firstly a two-parameter family of solutions where the configuration $\mathcal{A}$ is the positive half of $BC_{n}$ and multiplicities are Weyl-invariant. We obtain solutions with many parameters by taking special restrictions of these solutions (cf. \cite{Misha&Veslov 2008}, \cite{Misha & Maali}).      

We also apply these solutions in order to construct $\mathcal{N}=4$ supersymmetric mechanical systems. Relations of such mechanical systems with WDVV equations were known since \cite{Wyllard 2000} and \cite{Belluci+G+Latini 2005}.  
Trigonometric solutions of WDVV equations were used to construct $\mathcal{N}=4$ supersymmetric Hamiltonians in \cite{AF}. This gave, in particular, supersymmetric version of quantum Calogero--Moser--Sutherland system of type $BC_n$ with two independent coupling parameters. Thus we extend this Hamiltonian into multiparameter family. 
Other $\mathcal{N}=4$ supersymmetric systems related with Calogero--Moser--Sutherland system with extra spin variables were obtained in \cite{FIL} using the superfield approach.

It would be interesting to see whether constructed solutions of WDVV equations admit extra Frobenius manifolds structures or further links to geometry as it happens for root systems solutions with non-zero cubic terms $Q$ depending on an extra variable. 
Possible elliptic generalizations may also be interesting (see \cite{Riley+Strachan 2006}, \cite{BMMM 2007} and \cite{Strachan 2010} for elliptic solutions of WDVV equations).

\section{Metric for a family of $BC_n$ type configurations}\label{metricsection}
We are going to present solution $F$ to the equations (\ref{GDVV1}) for a suitable vector field $e.$ This solution is related to $BC_{n}$ root system with prescribed multiplicities of the root vectors.
Let $\underline{m}=(m_1, \dots, m_n)\in \mathbb{C}^{n}$ and let
$r,s,q \in \mathbb{C}.$ 
Let $BC_n(r, s,q ; \underline{m}) \subset \mathbb{C}^n$ be the following configuration of vectors $\alpha $ with corresponding multiplicities $c_\alpha$: 
\begin{align*}
e_i, \quad  \text{with multiplicities} \quad r m_i,\quad 1 \leq i \leq n,\\
 2e_i, \quad \text{with multiplicities} \quad s m_i + \frac{1}{2} qm_i(m_i -1),\quad 1 \leq i \leq n,\\
 e_i \pm e_j, \quad \text{with multiplicities} \quad qm_i m_j, \quad 1\leq i < j \leq n,
\end{align*}
where $e_1, \ldots, e_n$ is the standard basis in ${\mathbb C}^n$. 
If all the multiplicities $m_{i}=1$ then the configuration reduces to the configuration $BC_{n}(r,s,q)$ which is a positive half of the root system $BC_{n}$ with an invariant collection of multiplicities $r,s,q.$

Let us consider the function $F$  given by
\begin{align}\label{prepotential}
F= \sum_{\alpha \in BC_n(r, s,q ; \underline{m})} c_\alpha f((\alpha,x)),
\end{align}
where $f$ is given by (\ref{f(z)}). 
More explicitly the function $F$ can be written as follows:
\begin{align}\label{restrictedpotential}
F=  \sum_{i=1}^n r m_i f( x_i) + \sum_{i=1}^n ( s m_i + \frac{1}{2}qm_i(m_i -1)) f(2x_i)+  \sum_{i <j}^n qm_i m_j  f (x_i \pm  x_j).
\end{align}
Let us now define the matrix (\ref{Bmatrix}) by taking 
\begin{equation}\label{Ak}
 A_k = \sinh 2 x_k,
\end{equation}
where $k=1,...,n.$
This choice is motivated by \cite{Hoevenaars+Martini 2003} where a solution of WDVV equation (\ref{GDVV1}) for the root system $B_{n}$ was obtained.

Let us also define the following functions: 
\begin{align*}
b_{ij} = 
\begin{dcases}
\coth( x_i + x_j) + \coth( x_i - x_j) , \quad 1 \leq i \neq j \leq n,\\
0, \quad i=j,
\end{dcases}
\end{align*}
and 
\begin{equation*}
b_i = \coth x_i , \quad \quad \widetilde{b}_i= \coth 2 x_i, \quad i=1, \dots, n. 
\end{equation*}
\begin{lemma}\label{derevative of F}
We have the following expression for the third order derivatives of $F$:
\begin{align}\label{restrictedprepderi}
F_{klt} &= r  m_k b_k \delta_{kl} \delta_{lt}+ 4( 2 s m_k +  qm_k(m_k-1) ) \widetilde{b}_k  \delta_{kl} \delta_{lt}  + q \delta_{kl} \delta_{lt} \sum_{\substack{ j=1 \\ j \neq k}}^n m_j m_k b_{kj}  \\
&+ q  m_t m_k b_{tk}\delta_{kl} + q  m_l m_k b_{lk}\delta_{kt} + q  m_k m_l b_{kl}\delta_{lt} \nonumber,
\end{align}
where $k,l,t=1,...,n,$ and $\delta$ is the Kronecker symbol.
\end{lemma}
\begin{proof}
We note that the first two terms in \eqref{restrictedprepderi} are obtained from the first two terms in formula \eqref{restrictedpotential}. The last term in \eqref{restrictedpotential} contributes the following sum in $F_{klt}$:
\begin{align}\label{newterm}
q\sum_{i < j}^n m_i m_j \Big( (\delta_{ki}+ \delta_{kj})(\delta_{li}+ \delta_{lj})(\delta_{ti}+ \delta_{tj}) \coth(x_i + x_j) \\+ (\delta_{ki}- \delta_{kj})(\delta_{li}- \delta_{lj})(\delta_{ti}- \delta_{tj}) \coth(x_i - x_j)\Big) \nonumber.
\end{align}
We rearrange some of the terms in \eqref{newterm} as follows:
\begin{align}\label{newterm1}
q\sum_{i < j}^n m_i m_j ( \delta_{ki} \delta_{li} \delta_{ti} + \delta_{kj} \delta_{lj} \delta_{tj} ) \coth(x_i + x_j) =q \sum_{\substack{ j=1 \\ j \neq k}}^n m_j m_k \delta_{kl} \delta_{lt} \coth(x_k + x_j),
\end{align}
and
\begin{align}\label{newterm2}
q\sum_{i < j}^n m_i m_j ( \delta_{ki} \delta_{li} \delta_{ti} - \delta_{kj} \delta_{lj} \delta_{tj} ) \coth(x_i - x_j) = q\sum_{\substack{ j=1 \\ j \neq k}}^n m_j m_k \delta_{kl} \delta_{lt}  \coth(x_k - x_j). 
\end{align}
The sum of expressions \eqref{newterm1} and \eqref{newterm2} equals the third term in \eqref{restrictedprepderi}. Further on, let us collect the following terms from \eqref{newterm}:
\begin{align}\label{newterm3}
q\sum_{i < j}^n m_i m_j \big( \delta_{ki} \delta_{li} \delta_{tj} +\delta_{kj} \delta_{lj} \delta_{ti}  \big) \coth(x_i + x_j) =
\begin{dcases}
q  m_t m_k \delta_{kl}\coth(x_t + x_k), \quad t \neq k, \\
0, \quad t=k,
\end{dcases}
\end{align}
and
\begin{align}\label{newterm4}
q\sum_{i < j}^n m_i m_j \big( \delta_{ki} \delta_{li} \delta_{tj} -\delta_{kj} \delta_{lj} \delta_{ti}  \big) \coth(x_j - x_i) =
\begin{dcases}
 q m_t m_k \delta_{kl} \coth(x_t - x_k), \quad t \neq k ,\\
 0, \quad t=k.
\end{dcases}
\end{align}
The sum of the terms \eqref{newterm3} and \eqref{newterm4} equals to $q  m_t m_k b_{tk}\delta_{kl} $. Similarly, the sum of the terms 
\begin{align*}
q\sum_{i < j}^n m_i m_j \big( \delta_{ki} \delta_{lj} \delta_{tj} +\delta_{kj} \delta_{li} \delta_{ti}  \big) \coth(x_i + x_j) =
\begin{dcases}
q  m_k m_l \delta_{lt} \coth(x_k + x_l), \quad k \neq l, \\
0, \quad k=l,
\end{dcases}
\end{align*}
and
\begin{align*}
q\sum_{i < j}^n m_i m_j \big( \delta_{ki} \delta_{lj} \delta_{tj} -\delta_{kj} \delta_{li} \delta_{ti}  \big) \coth(x_i - x_j) =
\begin{dcases}
 q m_k m_l \delta_{lt} \coth(x_k - x_l), \quad k \neq l ,\\
 0, \quad k=l
\end{dcases}
\end{align*}
equals $q m_k m_l b_{kl}\delta_{lt}$. Finally, the sum of the following terms
\begin{align*}
q\sum_{i < j}^n m_i m_j \big( \delta_{ki} \delta_{lj} \delta_{ti} +\delta_{kj} \delta_{li} \delta_{tj}  \big) \coth(x_i + x_j) =
\begin{dcases}
q  m_l m_k \delta_{kt} \coth(x_l + x_k), \quad l \neq k, \\
0, \quad l=k,
\end{dcases}
\end{align*}
and
\begin{align*}
q\sum_{i < j}^n m_i m_j \big( \delta_{ki} \delta_{lj} \delta_{ti} -\delta_{kj} \delta_{li} \delta_{tj}  \big) \coth(x_j - x_i) =
\begin{dcases}
 q m_{l} m_{k} \delta_{kt} \coth(x_l - x_k), \quad k \neq l ,\\
 0, \quad l=k
\end{dcases}
\end{align*}
equals $q m_l m_k b_{lk} \delta_{kt}$. The statement follows. 
\end{proof}
\begin{lemma}\label{lemmaiden} We have the following identities:
\begin{align}\label{iden1}
A_k b_{kj}  + A_j b_{jk} = 2 ( \cosh 2 x_k + \cosh 2 x_j) , \quad  1 \leq j \neq k  \leq n,
\end{align} 
and 
\begin{align}\label{iden2}
A_k b_{jk}  + A_j b_{kj} =0, \quad j, k =1, \dots, n,
\end{align}
where $A_k$ is given by (\ref{Ak}).
\end{lemma}

\begin{proof}
We have 
\begin{align*}
A_k b_{kj} + A_j b_{jk} = \frac{A_k \sinh2x_k - A_j \sinh2x_j}{\sinh(x_k+ x_j) \sinh(x_k -x_j)}= \frac{\cosh4x_k - \cosh4x_j}{\cosh2x_k - \cosh2x_j},
\end{align*}
which implies the first formula \eqref{iden1}. Identity \eqref{iden2} follows similarly. 
\end{proof}
Now we show that the matrix $B$ is diagonal. Moreover, it is proportional to a constant diagonal matrix under a particular restriction on the parameters $r, q, s$. 

\begin{proposition}\label{B undermultconditions} 
The matrix $B=B(x)$ with the matrix entries
$$B_{lt}=\sum_{k=1}^{n} A_{k} F_{klt},\quad l,t =1,\dots,n$$
is diagonal. Furthermore, if the multiplicities $r$, $q$, $s$ and $\underline{m}$ satisfy the relation
\begin{align}\label{linearrelationmult}
r= -8s -2q (N-2),
\end{align}
where $N=\sum_{k=1}^{n}m_{k},$
then the matrix $B$ takes the form
\begin{align}\label{Bentries}
B_{lt} = m_l h(x) \delta_{lt},
\end{align}
where $h(x)= 2q \sum_{k=1}^n m_k \cosh 2x_k + r $.
\end{proposition}
\begin{proof}
It follows by Lemma \ref{derevative of F} that for $l \neq t $
\begin{align*}
B_{lt} = q m_l m_t (A_{l} b_{tl}  + A_{t} b_{lt} ), 
\end{align*}
which is equal to zero by Lemma \ref{lemmaiden}.

Let us now consider the diagonal entries of $B$. We have by Lemmas \ref{derevative of F}, \ref{lemmaiden}
\begin{align*}
B_{ll} &= r m_l A_l b_l + 4 ( 2 s m_l + q m_l (m_l -1) ) A_l \tilde{b}_l + q \sum_{\substack{ k=1 \\ k \neq l}}^n m_k m_l (A_l b_{lk} + A_k b_{kl}) \\
&= 2 r  m_l \cosh^2 x_l + 4( 2 s m_l + q m_l(m_l -1) ) \cosh 2 x_l  + 2q \sum_{\substack{ k=1 \\ k \neq l}}^n m_k m_l ( \cosh 2 x_k + \cosh 2 x_l).
\end{align*}
Then

\begin{align*}
B_{ll}&= 2 r  m_l \cosh^2 x_l + 4( 2 s m_l + qm_l(m_l -1) ) \cosh 2 x_l  + 2q (N-2m_l) m_l \cosh 2x_l \\&+ 2q m_l \sum_{k=1}^n m_k \cosh 2 x_k\\
&= m_l \Big(  \big( r + 8s + 2q (N-2) \big) \cosh 2 x_l + 2q \sum_{k=1}^n m_k \cosh 2x_k + r \Big),
\end{align*}
which implies the statement. 
\end{proof}
Below summation over repeated indices will be assumed. Let us now assume that multiplicities $m_{i}=1$ for all $i=1,...,n.$ For any vector $v=(v_1, \dots, v_n) \in V $ let us introduce the vector field $\partial_v = v_i \partial_{x_i} \in TV$. For any $u=(u_1, \dots, u_n) \in V$ we define the following multiplication on the tangent plane $T_{x}V$ for generic $x \in V \colon$
\begin{equation} \label{algebra}
\partial_{u} \ast \partial_{v}= F_{ijk} u_{i}v_{j} \partial_{x_{k}}.
\end{equation}
Note that multiplication \eqref{algebra} defines a commutative algebra on $T_xV$.
 The following theorem takes place.
\begin{theorem} \label{assoiativity&WDVV} 
Suppose that parameters $r$, $s$ and $q$ satisfy the linear relation (\ref{linearrelationmult}).
Then function 
\begin{align}\label{prepBCN}
F=r \sum_{i=1}^{n}f(x_{i})+s\sum_{i=1}^{n}f(2x_{i})+q\sum_{1\leq i<j\leq n}\big(f(x_{i}+x_{j})+f(x_{i}-x_{j})\big)
\end{align}
satisfies WDVV equations \eqref{GDVV1} where $B$ is determined by (\ref{Bmatrix}) and (\ref{Ak}). Also, multiplication \eqref{algebra} is associative. 
\end{theorem}
\begin{proof}
It has been shown in \cite{AF} that the function \eqref{prepBCN} satisfies the following system of equations if the linear relation \eqref{linearrelationmult} holds:
\begin{align}\label{FF=FF}
F_i F_j = F_j F_i, 
\end{align}
for all $i,j=1, \dots, n$. It then follows from Proposition \ref{B undermultconditions} that conditions \eqref{FF=FF} are equivalent to WDVV equations (\ref{GDVV1}) since the matirx $B$ is proportional to the identity matrix.
Also it is easy to see that associativity of the multiplication  \eqref{algebra} is equivalent to the relation (\ref{FF=FF}).
\end{proof}
In the remaining part of the paper we prove generalization of Theorem \ref{assoiativity&WDVV} to the configuration $BC_n(q, r, s ; \underline{m}),$ that is to the case of arbitrary multiplicities $m_{i}.$ This generalization can be formulated as follows. 
\begin{theorem} \label{restrictFandWDVV}
Suppose parameters $r,s,q$ and $\underline{m}$ satisfy the relation
\begin{align}\label{rsqrelation}
r=-8s-2q(N-2),
\end{align}
 where $N=\sum_{i=1}^{n}m_{i}.$ 
Then prepotential (\ref{restrictedpotential}) satisfies WDVV equations (\ref{GDVV1}) where $B=\sum_{i=1}^{n}\sinh 2x_{i}{F_{i}}.$
\end{theorem}

\begin{remark}
Theorem \ref{restrictFandWDVV} generalizes Theorem 2.3 from \cite{Hoevenaars+Martini 2003}. In this case we have all $m_{i}=1$ and $s=0.$ Then putting $q=1$ we get the standard $B_{N}$ root system and the condition (\ref{rsqrelation}) reduces to $r=-2(N-2)$ which is the multiplicity of the short root of $B_{N}$ root system considered in \cite{Hoevenaars+Martini 2003}. 
\end{remark}
\begin{remark}
In the rational limit solutions (\ref{restrictedpotential}) of WDVV equations reduce to $B_{n}$ family of $\vee$-system found in \cite{Chalykh+ Veselov 2001}.
\end{remark}

\section{Proof through restrictions}

Let $\mathcal{A}$ be the configuration $\mathcal{A}=BC_N(r,s,q)\subset W \cong \mathbb{C}^N, N \in \mathbb{N}$.
Let $e_1,\dots,e_N$ be the standard basis of $W.$ Let $(\cdot,\cdot)$ be the standard inner product which is defined by 
$$(x,y)= \sum_{i=1}^{N}x_{i}y_{i},$$
where $x=(x_{1},\dots,x_{N}),y=(y_{1},\dots,y_{N}) \in W.$ 
Let $n\in\mathbb{N}$ and $\underline{m}=(m_1, \dots, m_n)$ with $m_i \in \mathbb{N}$ such that $\sum_{i=1}^n m_i = N$. 
Let us consider subsystem $\mathcal{B}\subset \mathcal{A}$ as follows:
$$\mathcal{B}=\{e_{\sum_{j=1}^{i-1}m_{j}+k}-e_{\sum_{j=1}^{i-1}m_{j}+l},\quad 1\leq k<l\leq m_{i}, i=1,\dots,n\}.$$
Now let us consider the corresponding subspace of $W$ of dimension $n$ given by
$$W_{\mathcal{B}}=\{x\in W:(\beta,x)=0, \forall \beta\in \mathcal{B} \}.$$
More explicitly, vectors $x=(x_{1},\dots,x_{N})\in W_{\mathcal{B}}$ satisfy conditions:
\begin{numcases}{}
x_{1}=\dots =x_{m_{1}},\nonumber\\
x_{m_{1}+1}=\dots =x_{m_{1}+m_{2}},\nonumber \\
\vdots \nonumber \\
x_{m_{1}+m_{2}+\dots+m_{n-1}+1}=\dots =x_{N}.\label{W_B}
\end{numcases}
For any vector $v=(v_1, \dots, v_N) \in W $ let us define the vector field $\partial_v = v_i \partial_{x_i} \in TW$. For any $u=(u_1, \dots, u_N) \in W$ we define the following multiplication on the tangent plane $T_{x}W$ for generic $x \in W \colon$
\begin{equation} \label{algebra on TW}
\partial_{u} \ast \partial_{v}=u_{i}v_{j} F_{ijk} \partial_{x_{k}},
\end{equation}
where the function $F$ is given by
\begin{align}\label{prepBCN 2}
F= \sum_{\alpha \in \mathcal{A}} c_\alpha f((\alpha,x)). 
\end{align}
Assume that parameters $r,s,q$ and $\underline{m}$ satisfy the relation
$r=-8s-2q(N-2).$ Then multiplication (\ref{algebra on TW}) is associative by Theorem \ref{assoiativity&WDVV} (applied with $n=N$). 
Note that function (\ref{prepBCN 2}) satisfies
$$F_{ijk}=\sum_{\alpha\in\mathcal{A}}c_{\alpha}\alpha_{i}\alpha_{j}\alpha_{k}\coth(\alpha,x),$$
hence multiplication \eqref{algebra on TW} can be expressed as follows:
\begin{equation}\label{multiplicationformula}
\partial_{u}\ast \partial_{v}= \sum_{\alpha \in \mathcal{A}}c_{\alpha} (\alpha,u)(\alpha,v)\coth (\alpha,x)\partial_{\alpha}.
\end{equation}
If we identify $W$ with $T_{x}W \cong W,$ then multiplication (\ref{multiplicationformula}) takes the form
\begin{equation}\label{a*b}
u\ast v= \sum_{\alpha \in \mathcal{A}}c_{\alpha} (\alpha,u)(\alpha,v)\coth (\alpha,x)\alpha. 
\end{equation}
Define $M_{\mathcal{B}}=W_{\mathcal{B}} \setminus \bigcup_{\alpha \in \mathcal{A}\setminus \mathcal{B}} \Pi_{\alpha},$ 
where $ \Pi_{\alpha}= \{ x \in W \colon (\alpha,x)=0 \} .$
Consider now a point $x_{0} \in M_{\mathcal{B}}$ and two tangent vectors $u_{0}, v_{0} \in T_{x_{0}}M_{\mathcal{B}}.$ We extend vectors $u_{0}$ and $v_{0}$ to two local analytic vector fields $u(x), v(x)$ in the neighbourhood $U$ of $x_{0}$ that are tangent to the subspace $W_{\mathcal{B}} $ at any point $x \in M_{\mathcal{B}} \cap U$ such that $u_{0}=u(x_{0})$ and  $v_{0}=v(x_{0})$. Now we want to study the limit of $u(x) \ast v(x)$ when $x$ tends to $x_{0}.$
The limit may have singularities at $x \in W_{\mathcal{B}} $ as $\coth (\alpha,x)$ with $\alpha \in \mathcal{B} $ is not defined for such $x.$ Also we note that outside $W_{\mathcal{B}}$ we have a well-defined multiplication $u(x) \ast v(x).$ Similarly to the rational case considered in \cite{Feigin+Veselov 2007} and trigonometric case with extra variable \cite{Misha & Maali} the following lemma holds.
\begin{lemma}\label{welldefalgebra}
The product $u(x) \ast v(x)$ has a limit when  $x$ tends to $x_{0}\in  M_{\mathcal{B}} $ given by
\begin{equation} \label{productlimit}
u_{0} \ast v_{0} =\sum_{\alpha \in \mathcal{A}\setminus \mathcal{B}} c_{\alpha}(\alpha,u_{0})(\alpha,v_{0})\coth (\alpha ,x_{0})\alpha.
\end{equation}
In particular $u_{0} \ast v_{0}$ is determined by $u_{0}$ and $v_{0}$ only. 
\end{lemma}
Now we are going to show that the product $u_{0} \ast v_{0}$ belongs to $T_{x_{0}}M_{\mathcal{B}}$. We will need the following lemma (cf. \cite{Feigin+Veselov 2007}, \cite{Misha & Maali}).
\begin{lemma}\label{closedidentity}
Let $\alpha \in \mathcal{A}.$ Let $x\in \Pi_{\alpha}$ be generic. Then the identity 
\begin{equation}\label{rootsysidentity}
 \sum_{\substack{\beta \in \mathcal{A}\\ \beta \nsim \alpha}} c_{\beta}(\alpha,\beta)\coth (\beta,x) B_{\alpha,\beta}(a,b)\alpha \wedge\beta=0
\end{equation} 
holds for all $a,b\in V$ provided that $(\alpha,x)=0,$ where
$B_{\alpha,\beta}(a,b)=(\alpha,a)(\beta,b)-(\alpha,b)(\beta,a)$ and $\alpha \wedge\beta=\alpha\otimes \beta-\beta\otimes\alpha.$   
\end{lemma}
\begin{proof}
For any $\beta \in \mathcal{A}$ such that $\beta \nsim \alpha$ let $\gamma=s_{\alpha}\beta$, where $s_{\alpha}$ is the orthogonal reflection about $\Pi_{\alpha}.$
Note that $\coth (\gamma,x)= \coth(\beta,x)$ at $(\alpha,x)=0.$
Also note that
\begin{align*}
(\alpha,\gamma)=-(\alpha,\beta),\quad B_{\alpha,\gamma}(a,b)=B_{\alpha,\beta}(a,b),\quad \alpha\wedge \gamma = \alpha\wedge\beta.
\end{align*}
We have that either $\gamma$ or $-\gamma$ is an element of $\mathcal{A}.$ Suppose firstly that $\gamma\in \mathcal{A}.$
Then 
\begin{align}\label{pair of terms}
c_{\beta}(\alpha,\beta)\coth (\beta,x) B_{\alpha,\beta}(a,b)\alpha \wedge\beta +c_{\gamma}(\alpha,\gamma)\coth (\gamma,x)B_{\alpha,\gamma}(a,b)\alpha\wedge \gamma=0
\end{align}
at $(\alpha,x)=0$ since multiplicities are $B_{N}$-invariant. If one replaces $\gamma$ with $-\gamma$ then (\ref{pair of terms}) holds as well.  
\end{proof}
\begin{proposition}\label{closedalgebra}
Let $u,v \in T_{x}M_{\mathcal{B}}$ where $x \in M_{\mathcal{B}}.$ Then $u \ast v \in T_{x}M_{\mathcal{B}}.$ 
\end{proposition}
\begin{proof}
By Lemma \ref{welldefalgebra} it is enough to show that 
\begin{equation} \label{closedcondition}
\sum_{\beta \in \mathcal{A}\setminus \mathcal{B}}c_{\beta}(\beta,u) (\beta,v) (\alpha,\beta) \coth (\beta,x)=0
\end{equation}
for all $\alpha \in \mathcal{B}.$
Assume firstly that $W_{\mathcal{B}}$ has codimension $1.$ By Lemma \ref{closedidentity} we get
\begin{equation}\label{idenforarbitrary}
\sum_{\substack{\beta \in \mathcal{A}\\
                 \beta \nsim \alpha}} c_{\beta}(\alpha,\beta)\coth (\beta,x)\bigg((\alpha,a)(\beta,b)-(\alpha,b)(\beta,a)\bigg)\bigg((\alpha,y)(\beta,z)-(\alpha,z)(\beta,y)\bigg)=0
\end{equation} 
for any $a,b,y,z \in V $ and generic $x\in \Pi_{\alpha}.$
Assume that $a,y \notin \Pi_{\alpha}$ and let $b=u \in \Pi_{\alpha}$ and $z=v \in \Pi_{\alpha}$. Then $(\alpha,b)=(\alpha,z)=0$ and relation (\ref{idenforarbitrary}) implies that
$$\sum_{\substack{\beta \in \mathcal{A}\\
                 \beta \nsim \alpha}} c_{\beta}(\alpha,\beta)(\beta,u)(\beta,v)\coth (\beta,x)=0.$$ 
As $W_{\mathcal{B}}$ has codimension $1$ the relation $ \beta \nsim \alpha$ is equivalent to $\beta \in \mathcal{A}\setminus \mathcal{B}$ and lemma follows.

Let us now suppose that $W_{\mathcal{B}}$ has codimension $2.$ Let $\alpha, \gamma \in  \mathcal{B},  \alpha \nsim \gamma.$ By the above arguments for generic
 $x \in \Pi_{\alpha}$ and $u,v \in T_{x}\Pi_{\alpha} $, we have $u \ast v \in T_{x}\Pi_{\alpha}.$
Similarly if $x \in \Pi_{\gamma}$ is generic and $u,v \in T_{x}\Pi_{\gamma} $, then $u \ast v \in T_{x}\Pi_{\gamma}.$ By Lemma \ref{welldefalgebra}, $u \ast v$ exists for $x \in M_{\mathcal{B}}$ and $u,v\in T_{x}M_{\mathcal{B}}.$ It follows that for any $x \in M_{\mathcal{B}}$ we have $u \ast v \in T_{x}M_{\mathcal{B}},$
which proves the statement for the case when $W_{\mathcal{B}}$ has codimension $2.$ General $W_{\mathcal{B}}$ is dealt with similarly.
\end{proof}
Consider now the orthogonal decomposition
\begin{equation} \label{orthogonaldecompsition}
W= W_{ \mathcal{B}} \oplus W_{ \mathcal{B}}^\bot 
\end{equation}
with respect to the standard inner product.
Any $\alpha\in W$ can be written as 
\begin{equation} \label{vectordecomp}
\alpha= \widetilde{\alpha} + w ,
\end{equation}
where $ \widetilde{\alpha} \in  W_{ \mathcal{B}}$ is the orthogonal projection of vector $\alpha$ to $W_{ \mathcal{B}}$ and $w \in  W_{ \mathcal{B}}^\bot. $
For any $x_{0}\in M_{\mathcal{B}}$ and $u,v\in T_{x_{0}}M_{\mathcal{B}}$ one can represent product $u\ast v$ as 
\begin{equation} \label{restrictedproduct2}
u \ast v=\sum_{\alpha \in \mathcal{A}\setminus \mathcal{B}} c_{\alpha}(\alpha,u)(\alpha,v) \coth (\alpha ,x_{0})\widetilde{\alpha}
\end{equation}
by Proposition  \ref{closedalgebra}.
Hence, we have 
\begin{equation}\label{restrictedproduct3}
\partial_{u}\ast \partial_{v}=\sum_{\alpha \in \mathcal{A}\setminus \mathcal{B}} c_{\alpha}(\alpha,u)(\alpha,v) \coth (\alpha,x) \partial_{\widetilde{\alpha}}.
\end{equation}

Let us define vectors  $f_{i}, 1	\leq i 	\leq n $ by
\begin{equation}\label{resbasisf_i}
f_{i}=\sum_{j=1}^{m_{i}}e_{\sum_{s=1}^{i-1}m_{s}+j}.
\end{equation}  
These vectors form a basis for $W_{\mathcal{B}}.$ 

The following lemma gives the general formula for the orthogonal projection of any vector $u \in W$ to $W_{\mathcal{B}}.$
\begin{lemma}\label{vector u decomposition}
Let $u=\sum_{i=1}^{N}u_{i}e_{i}\in W.$ Then the projection $\widetilde{u}$ has the form 
\begin{equation}\label{orthogonal projection formula 2}
\widetilde{u}=\Bigg(\underbrace{\frac{1}{m_{1}}\displaystyle\sum_{i=1}^{m_{1}}u_{i},\dots,\frac{1}{m_{1}}\sum_{i=1}^{m_{1}}u_{i}}_{m_{1}},\dots , \underbrace{\frac{1}{m_{n}}\sum_{i=1}^{m_{n}}u_{\sum_{s=1}^{n-1}m_{s}+i},\dots,\frac{1}{m_{n}}\sum_{i=1}^{m_{n}}u_{\sum_{s=1}^{n-1}m_{s}+i}}_{m_{n}}\Bigg).
\end{equation}
\end{lemma}
Let us now project  $\mathcal{A}$ to the subspace $W_{\mathcal{B}}$. 
Notice that by Lemma \ref{vector u decomposition}
$$\widetilde{e}_{\sum_{s=1}^{i-1}m_{s}+1}=\dots=\widetilde{e}_{\sum_{s=1}^{i}m_{s}}=m_{i}^{-1}f_{i},\quad i=1,\dots,n.$$
\noindent Let us denote the projected system as 
$\widetilde{\mathcal{A}}=BC_n(q, r, s ; \underline{m}) \subset \mathbb{C}^n.$ It consists of vectors $\alpha $ with the corresponding multiplicities $c_{\alpha}$ given as follows:\\ 
\begin{align*}
\widehat{f}_{i}=m_{i}^{-1}{f_{i}}, \quad  \text{with multiplicities} \quad r m_i,\quad 1 \leq i \leq n,\\
 2\widehat{f}_{i}=2 m_{i}^{-1}{f_{i}}, \quad \text{with multiplicities} \quad s m_i + \frac{1}{2} qm_i(m_i -1),\quad 1 \leq i \leq n,\\
 \widehat{f}_{i} \pm \widehat{f}_{j}=m_{i}^{-1}{f_{i}}\pm m_{j}^{-1}{f_{j}}, \quad \text{with multiplicities} \quad qm_i m_j, \quad 1\leq i < j \leq n.
\end{align*}
By Lemma \ref{vector u decomposition}, for any $\alpha\in W,$ its orthogonal projection has the form
$$\widetilde{\alpha}=\sum_{k=1}^{n}\widetilde{\alpha}_{k}f_{k},$$
where the basis $f_{k}$ is given by (\ref{resbasisf_i}) and 
\begin{equation}\label{tilde alpha component}
\widetilde{\alpha}_{k}=\frac{(\widetilde{\alpha},f_{k})}{(f_{k},f_{k})}=\frac{(\alpha,f_{k})}{m_{k}}.
\end{equation}
Let us define 
\begin{equation}\label{tilde F}
\widetilde{F}(\widetilde{x})=\displaystyle \sum_{\gamma\in \widetilde{\mathcal{A}}}c_{\gamma}f((\gamma,\widetilde{x})),
\end{equation}
where 
\begin{equation}\label{tilde x}
\widetilde{x}=\displaystyle\sum_{i=1}^{n}\widetilde{x}_{i}f_{i}\in W_{\mathcal{B}}.
\end{equation}
Note that function (\ref{tilde F}) can also be represented as 
$$\widetilde{F}(\widetilde{x})=\sum_{\alpha\in\mathcal{A}\setminus\mathcal{B}}c_{\alpha}f((\alpha,\widetilde{x})).$$
Let $\widetilde{F}_{i}$ be the $n\times n$ matrix constructed from the third-order partial derivatives of the function $\widetilde{F},$ that is
\begin{align*}
(\widetilde{F}_{i})_{jk}=\widetilde{F}_{ijk} = \frac{\partial^3 \widetilde{F}}{\partial \widetilde{x}_i \partial \widetilde{x}_j \partial \widetilde{x}_k},
\end{align*}
$i,j,k=1,\dots,n.$

The following lemma gives another way to represent multiplication (\ref{restrictedproduct3}) on $W_{\mathcal{B}}.$
\begin{lemma}\label{product on W_B}
The multiplication (\ref{restrictedproduct3}) takes the form
$$\partial_{f_{i}}\ast \partial_{f_{j}}=\sum_{k=1}^{n}m_{k}^{-1}\widetilde{F}_{ijk}\partial_{f_{k}}, \quad i,j=1,...,n.$$ 
\end{lemma}
\begin{proof}
We rearrange $\partial_{\widetilde{\alpha}}$ in the right hand side of (\ref{restrictedproduct3}) as
$$\partial_{\widetilde{\alpha}}=\sum_{k=1}^{n}\widetilde{\alpha}_{k}\partial_{f_{k}}=\sum_{k=1}^{n}m_{k}^{-1}(\alpha,f_{k})\partial_{f_{k}}$$
by (\ref{tilde alpha component}).
Therefore the multiplication (\ref{restrictedproduct3})
can be rewritten as
\begin{align*}
\partial_{f_{i}}\ast \partial_{f_{j}}&=\sum_{\alpha\in \mathcal{A}\setminus\mathcal{B}}\sum_{k=1}^{n} c_{\alpha}m_{k}^{-1}(\alpha,f_{i})(\alpha,f_{j})(\alpha,f_{k})\coth(\alpha,\widetilde{x}) \partial_{f_{k}}\\
&=\sum_{k=1}^{n} m_{k}^{-1}\widetilde{F}_{ijk}\partial_{f_{k}},\quad i,j=1,...,n,
\end{align*}
as required.
\end{proof}
Let $H_{\mathcal{B}}$ be the matrix of the restriction of the standard inner product on $W_{\mathcal{B}}$ in the basis $f_{1},\dots,f_{n}.$ That is 
\begin{equation}\label{H_B}
(H_{\mathcal{B}})_{lt}=(f_{l},f_{t})=m_{l}\delta_{lt}, \quad l,t=1,\dots,n.
\end{equation}
\begin{lemma}\label{restricted matrix H_B}
The matrix $H_{\mathcal{B}}$ can be written as a linear combination $$H_{\mathcal{B}}=\sum_{i=1}^{n}a_{i}{\widetilde{F}_{i}},$$ where functions $a_{i}$ are given by $a_{i}=h(\widetilde{x})^{-1}{\sinh 2\widetilde{x}_{i}},$ and $h(\widetilde{x})= 2q \sum_{k=1}^n m_k \cosh 2\widetilde{x}_k + r.$
\end{lemma}
\begin{proof}
By Proposition \ref{B undermultconditions} and (\ref{H_B}), we have
$H_{\mathcal{B}}=h(\widetilde{x})^{-1}{B(\widetilde{x})},$
where $B(\widetilde{x})=\displaystyle\sum_{i=1}^{n}(\sinh 2 \widetilde{x}_{i})\widetilde{F}_{i}.$ This implies the statement.
\end{proof}
The previous considerations allow us to prove the following theorem, which is a version of stated earlier Theorem \ref{restrictFandWDVV} where multiplicities $m_{i}$ do not have to be integer.
\begin{theorem} \label{restrictFandWDVV 2}
Let $\widetilde{\mathcal{A}}=BC_n(q, r, s ; \underline{m}) \subset \mathbb{C}^n.$ Suppose parameters $r,s,q$ and $\underline{m}$ satisfy the relation
\begin{equation}\label{linear conditions for restriction}
r=-8s-2q(N-2),
\end{equation}
where $N=\sum_{i=1}^{n}m_{i}.$ 
Then the prepotential
\begin{equation} \label{F_B}
\widetilde{F}=\sum_{\alpha \in \widetilde{\mathcal{A}}}c_{\alpha}f((\alpha,\widetilde{x})),\quad \widetilde{x}\in W_{\mathcal{B}}, 
\end{equation}
satisfies the WDVV equations 
\begin{equation}\label{GDVV1 dim=N}
\widetilde{F}_{i}B^{-1}\widetilde{F}_{j}=\widetilde{F}_{j}B^{-1}\widetilde{F}_{i},\quad i,j=1,...,n,
\end{equation}
where $B=\sum_{i=1}^{n}\sinh2\widetilde{x}_{i}\widetilde{F}_{i}.$

The corresponding associative multiplication has the form 
\begin{equation} 
u \ast v =\sum_{\widetilde{\alpha} \in \widetilde{\mathcal{A}}} c_{\widetilde{\alpha}}(\widetilde{\alpha},u)(\widetilde{\alpha},v) \coth (\widetilde{\alpha} ,\widetilde{x})\widetilde{\alpha},
\end{equation}
for any $u,v \in T_{\widetilde{x}}M_{\mathcal{B}}, \widetilde{x}\in M_{\mathcal{B}}, $ 
where $\widetilde{\alpha}$ is the orthogonal projection of $\alpha$ to $W_{\mathcal{B}}.$
\end{theorem}
\begin{proof}
Let us assume firstly that $m_{i}\in\mathbb{N}$ for any $i.$ Consider the multiplication
\begin{equation}\label{u*v}
 u\ast v= \sum_{\alpha \in BC_{N}(r,s,q)}c_{\alpha} (\alpha,u)(\alpha,v)\coth (\alpha,x)\alpha
 \end{equation}
 on the tangent space $T_{x}W$ for $x \in W.$ By Theorem \ref{assoiativity&WDVV} the multiplication (\ref{u*v}) is associative.
Now as $x$ tends to $\widetilde{x}\in M_{\mathcal{B}},$ by Lemmas \ref{welldefalgebra}, \ref{product on W_B} and Proposition \ref{closedalgebra} this product restricts to an associative product on the tangent space $T_{\widetilde{x}}M_{\mathcal{B}}$ which has the form 
\begin{equation}\label{thenewproduct}
\partial_{f_{i}}\ast \partial_{f_{j}}=\sum_{l=1}^{n}m_{l}^{-1}\widetilde{F}_{ijl}\partial_{f_{l}}.
\end{equation}
The associativity condition
\begin{align*}
(\partial_{f_{i}} \ast \partial_{f_{j}}) \ast \partial_{f_{k}}= \partial_{f_{i}} \ast ( \partial_{f_{j}} \ast \partial_{f_{k}}), 
\end{align*}
for any $i, j , k= 1, \dots, n,$ can be rearranged as
\begin{align*}
\sum_{l=1}^{n}m_{l}^{-1}\widetilde{F}_{ijl} \partial_{f_{l}} \ast \partial_{f_{k}}= \sum_{l=1}^{n}m_{l}^{-1}\widetilde{F}_{jkl} \partial_{f_{i}} \ast \partial_{f_{l}}.
\end{align*}
Hence, we have 
\begin{align}\label{L & R}
\sum_{l=1}^{n}m_{l}^{-1}\widetilde{F}_{ijl}\widetilde{F}_{lkp} =\sum_{l=1}^{n}m_{l}^{-1}\widetilde{F}_{jkl}\widetilde{F}_{ilp},
\end{align}
for any $i,j,k,p.$
In the matrix form we have 
$$\widetilde{F}_{i}H_{\mathcal{B}}^{-1}\widetilde{F}_{k}=\widetilde{F}_{k}H_{\mathcal{B}}^{-1}\widetilde{F}_{i}.$$
By Lemma \ref{restricted matrix H_B} we obtain relation (\ref{GDVV1 dim=N})
where 
$B=\sum_{i=1}^{n}\sinh2\widetilde{x}_{i}\widetilde{F}_{i}$
as required. This proves the theorem for the case when $m_{i}\in \mathbb{N}.$ Since $m_{i}$ can take arbitrary integer values the statement follows for general $m_{i}$ as well.
\end{proof}

\section{Application to supersymmetric mechanics}

Let us define coordinates $\widehat{x}=(\widehat{x}_1, \dots, \widehat{x}_n) \in \mathbb{C}^n$ by $\widehat{x}_i= m_i^{1/2} x_i$ ($1 \leq i \leq n$). Let $\Lambda$ be the $n \times n$ diagonal matrix $\Lambda= (m_i^{1/2}  \delta_{ij})_{i,j=1}^n$.
Let $F$ be given by formula  \eqref{restrictedpotential} with relation \eqref{linearrelationmult} on parameters $r,q,s$. By Proposition \ref{B undermultconditions} the  matrix $B$ can be represented as $B=h(x) \Lambda^2$.  Let us define a function $\widehat{F}(\widehat{x})$ such that $\widehat{F}(\widehat{x})=F(x(\widehat{x}))$. Consider the $n \times n$ matrices $\widehat{F}_k$ with entries
\begin{align}\label{tildederiv}
(\widehat{F}_k)_{lt}=\widehat{F}_{klt}= \frac{\partial^3 \widehat{F}}{\partial \widehat{x}_k \partial \widehat{x}_l \partial \widehat{x}_t},
\end{align}
$k,l,t=1,\dots,n.$
Note that $\widehat{F}_k=m_k^{-1/2} \Lambda^{-1} F_k \Lambda^{-1}$, where $F_k$ is the matrix with entries $(F_k)_{lt}=F_{klt}$. Let $\widehat{B}$ be the $n\times n$ matrix with entries $\widehat{B}_{ij}=h\delta_{ij}$, where $h=h(x(\widehat{x}))$ is given in Proposition \ref{B undermultconditions}.

\begin{proposition}\label{propoFtildeWDVV}
The metric $\widehat{B}$ can be represented as
\begin{align*}
\widehat{B}= \sum_{k=1}^n m_k^{1/2} \sinh( 2 m_k^{-1/2} \widehat{x}_k )\widehat{F}_{k},
\end{align*}
and the function $\widehat{F}$ satisfies generalised WDVV equations of the form
\begin{align}\label{FtildeWDVV}
\widehat{F}_i \widehat{F}_j = \widehat{F}_j \widehat{F}_i,
\end{align}
for all $i,j=1, \dots, n$. 
\end{proposition}

\begin{proof}
The first part of the statement is immediate by Proposition \ref{B undermultconditions}. Consider the left-hand-side of \eqref{FtildeWDVV}. We have
\begin{align}\label{FtildeWDVV1}
\widehat{F}_i \widehat{F}_j=  (m_i m_j)^{-1/2} \Lambda^{-1} F_i  \Lambda^{-2} F_j \Lambda^{-1}= h (m_i m_j)^{-1/2}  \Lambda^{-1} F_i  B^{-1} F_j \Lambda^{-1}.
\end{align}
It follows by Theorem \ref{restrictFandWDVV} that the right-hand-side of \eqref{FtildeWDVV1} is symmetric under the swap of $i$ and $j$, hence the statement follows. 
\end{proof}

Two different representations of $\mathcal{N}=4$ supersymmetry algebra were constructed in \cite{AF} (see also references therein). The corresponding supercharges depend on a prepotential of the form \eqref{trigonometric.solution} (with $Q=0$). This prepotential is assumed to satisfy equations of the form \eqref{FF=FF}. It follows by Proposition \ref{propoFtildeWDVV} that the function $\widehat{F}$ satisfies such type of equations, hence we obtain in this way two supersymmetric Hamiltonians $H_i$, $i=1,2$ for a family of $BC_n$ type configurations. We give these Hamiltonians in detail.
Let us consider the following configuration $\widehat{\mathcal{A}} \subset \mathbb{C}^n$ of vectors $\alpha$ with multiplicities $c_\alpha$:
\begin{align*}
 m_i^{-1/2} e_i, \quad \text{with multiplicities} \quad r m_i, \quad 1 \leq i \leq n, \\
  2 m_i^{-1/2} e_i, \quad \text{with multiplicities} \quad s m_i + \frac{1}{2} qm_i(m_i -1),\quad 1 \leq i \leq n, \\
 m_i^{-1/2} e_i \pm m_j^{-1/2} e_j, \quad \text{with multiplicities} \quad q m_i m_j, \quad 1 \leq i < j \leq n. 
\end{align*}

Consider $n$ (quantum) particles on a line with coordinates $\widehat{x}_j$ and momenta $p_j=- i  \partial_{\widehat{x}_j}$,  $j=1, \dots, n$ to each of which we associate four fermionic variables $ \psi^{aj}, \bar{\psi}_a^j , a=1, 2.$ 
These variables may be thought of as operators acting on wavefunctions which depend on bosonic and fermionic variables. 
Let $\epsilon_{ab}$ be the fully anti-symmetric tensors in two dimensions such that $ \epsilon_{12}= - \epsilon_{21}=1$.

Fermionic variables are assumed to satisfy the following (anti)-commutation relations ($a, b = 1, 2; j,k = 1, \dots, n$):
\begin{align*}
 \quad \{ \psi^{aj}, \bar{\psi}_b^k \} =-\frac{1}{2} \delta_{jk} \delta_{ab}, \quad \{\psi^{aj}, \psi^{bk}\}=\{  \bar{\psi}_a^j,  \bar{\psi}_b^k\}=0.
\end{align*}
We consider supercharges of the form 
\begin{align*}
Q^a& = -i\frac{\partial}{\partial \widehat{x}_r} \psi^{ar} +i \widehat{F}_{rjk} (\epsilon_{bc} \epsilon_{da} \psi^{br} \psi^{cj}  \bar{\psi}_d^{k} - \frac{1}{2} \psi^{ar} \delta_{jk}),\\
\bar{Q}_a &= -i\frac{\partial}{\partial \widehat{x}_l} \bar{\psi}_a^l + i \widehat{F}_{lmn}(\epsilon_{bd}\epsilon_{ac}  \bar{\psi}_d^l  \bar{\psi}_{b}^{m} \psi^{cn}  - \frac{1}{2} \bar{\psi}_a^l \delta_{nm} ).
\end{align*}
where $a =1,2$, $\widehat{F}_{ijk}$ is defined in \eqref{tildederiv}, and we assume summation over repeated indecies. Let $\Delta=\sum_{j=1}^n \partial_{\widehat{x}_j}^2$ be the Laplacian in $\mathbb{C}^n$. We have the following statement on supersymmetry algebra which follows from \cite{AF}. 

\begin{proposition}For all $a, b =1, 2$ the supercharges $Q^a$, $\bar{Q}_b$ generate  $\mathcal{N}=4$ supersymmetry algebra of the form
$$
\{Q^a, Q^b\} = \{\bar{Q}_a, \bar{Q}_b\} = 0, \quad \text{and}\quad   
\{Q^a, \bar{Q}_b  \}=-\frac12 H_1 \delta_{ab},
$$
 where the Hamiltonian $H_1$ is given by 
 \begin{align*}
H_1= - \Delta + \frac{1}{2} \sum_{\alpha \in \mathcal{\widehat{A}}} \frac{c_\alpha(\alpha, \alpha)^2}{\sinh^2(\alpha, \widehat{x})} + \frac{1}{4} \sum_{\alpha, \beta \in \mathcal{\widehat{A}}}c_\alpha c_\beta (\alpha, \alpha) (\beta, \beta) (\alpha, \beta) \coth(\alpha,\widehat{x}) \coth(\beta, \widehat{x}) + \Phi, 
\end{align*}
with the fermionic term
\begin{align}\label{phitrig}
\Phi=  \sum_{\alpha \in \mathcal{\widehat{A}}}  \frac{2 c_\alpha \alpha_i \alpha_j }{\sinh^2(\alpha, \widehat{x})}\big(\alpha_l \alpha_k \epsilon_{bc}\epsilon_{a d} \psi^{bi} \psi^{cj} \bar{\psi}_d^l \bar{\psi}_{a}^{k} +(\alpha, \alpha) \psi^{ai} \bar{\psi}_{a}^{j} \big).
\end{align}
\end{proposition}

The Hamiltonian $H_1$ is formally self-adjoint. Similar considerations (see \cite{AF}) yield not self-adjoint Hamiltonian of the form 
\begin{align*}
H_2= - \Delta + \sum_{\alpha \in \mathcal{\widehat{A}}}c_\alpha (\alpha, \alpha) \coth(\alpha,\widehat{x}) \partial_\alpha + \Phi
\end{align*}
with the same fermionic term $\Phi.$ 
In fact Hamiltonians $H_1$, $H_2$ satisfy gauge relation
%
$H_{2}=g H_1 g^{-1}$, 
%
where $g= \prod_{\alpha \in \mathcal{\widehat{A}}} \sinh^{\frac{c_\alpha}{2} (\alpha,\alpha)}(\alpha,\widehat{x})$.

\end{document}